\documentclass[preprint]{elsarticle}

\usepackage{array,xspace,multirow,hhline,tikz,colortbl,tabularx,booktabs,fixltx2e,amsmath,amssymb,amsfonts,amsthm}
\usepackage{algorithm}
\usepackage{algorithmic}
\usepackage{verbatim,ifthen}
\usepackage{enumitem}
\usepackage{pifont}
\usepackage{ifthen}
\usepackage{calrsfs,mathrsfs}
\usepackage{bbding,pifont}
\usepackage{pgflibraryshapes}
\usepackage{nicefrac}
\usepackage{tabularx}
\usepackage{mathtools}
\usepackage{pgflibraryshapes}
\usetikzlibrary{arrows}
\usepackage{centernot}
\usetikzlibrary{decorations.pathreplacing}
\usetikzlibrary{patterns}
\usepackage{pgflibraryshapes}
\usetikzlibrary{positioning,chains,fit,shapes,calc}

\definecolor{light-gray}{gray}{0.9}

	\newtheorem{lemma}{Lemma}%
	\newtheorem{theorem}{Theorem}%
	\newtheorem{corollary}{Corollary}%

	\newtheorem{definition}{Definition}

	\newcommand\eat[1]{}

	\usepackage{enumitem}
	\setenumerate[1]{label=\rm(\it{\roman{*}}\rm),ref=({\it\roman{*}}),leftmargin=*}
	\newlength{\wordlength}

	\newcommand{\midd}{\mathbin{:}}

	\newcommand{\strong}{SW\xspace}

	\algsetup{linenodelimiter=\,}
	\algsetup{linenosize=\tiny}
	\algsetup{indent=2em}

			\newcommand{\ml}[1][]{\ensuremath{\ifthenelse{\equal{#1}{}}{\mathit{ML}}{\mathit{ML}(#1)}}\xspace}
			\newcommand{\sml}[1][]{\ensuremath{\ifthenelse{\equal{#1}{}}{\mathit{SML}}{\mathit{SML}(#1)}}\xspace}
			\newcommand{\sd}[1][]{\ensuremath{\ifthenelse{\equal{#1}{}}{\mathit{SD}}{\mathit{SD}(#1)}}\xspace}
			\newcommand{\rsd}[1][]{\ensuremath{\ifthenelse{\equal{#1}{}}{\mathit{RSD}}{\mathit{RSD}(#1)}}\xspace}
			\newcommand{\rd}[1][]{\ensuremath{\ifthenelse{\equal{#1}{}}{\mathit{RD}}{\mathit{RD}(#1)}}\xspace}
			\newcommand{\st}[1][]{\ensuremath{\ifthenelse{\equal{#1}{}}{\mathit{ST}}{\mathit{ST}(#1)}}\xspace}
			\newcommand{\bd}[1][]{\ensuremath{\ifthenelse{\equal{#1}{}}{\mathit{BD}}{\mathit{BD}(#1)}}\xspace}
			\newcommand{\pc}[1][]{\ensuremath{\ifthenelse{\equal{#1}{}}{\mathit{PC}}{\mathit{PC}(#1)}}\xspace}
			\newcommand{\dl}[1][]{\ensuremath{\ifthenelse{\equal{#1}{}}{\mathit{DL}}{\mathit{DL}(#1)}}\xspace}
			\newcommand{\ul}[1][]{\ensuremath{\ifthenelse{\equal{#1}{}}{\mathit{UL}}{\mathit{UL}(#1)}}\xspace}

				\newcommand{\pref}{\ensuremath{\succsim}}
				\newcommand{\spref}{\ensuremath{\succ}}
				\newcommand{\indiff}{\ensuremath{\sim}}

\begin{document}

	\title{Characterizing SW-Efficiency\\ in the Social Choice Domain}


	\author{Haris Aziz\corref{cor1}} \ead{haris.aziz@nicta.com.au}
		\address{NICTA and UNSW, 2033 Sydney, Australia}



	\begin{abstract}
		Recently, Dogan, Dogan and Yildiz (2015) presented a new efficiency notion for the random assignment setting called SW (social welfare)-efficiency and characterized it. In this note, we generalize the characterization for the more general domain of randomized social choice. 
\end{abstract}

	\begin{keyword}
	 Social decision schemes
		\sep Social choice theory
		\sep Pareto optimal 
		\sep Social Welfare
			\sep Stochastic Dominance\\
			\emph{JEL}: C63, C70, C71, and C78
	\end{keyword}

\maketitle

\section{Introduction}

The \emph{random assignment setting} captures the scenario in which $n$ agents express preferences over $n$ objects and the outcome is a probabilistic assignment. For the the setting, two interesting efficiency notions are ex post efficiency and SD (stochastic dominance)-efficiency~\citep{AbSo03a,ABB14b,AMXY15a,BoMo01a,Cho12a,Mcle02a}. The assignment setting can be considered as a special case of voting where each discrete assignment can be viewed as a voting alternative~\citep{Aziz14c,ABBH12a,Carr10a}. 

Recently, \citet{DDY15a} presented a new notion of efficiency called \strong (social welfare)-efficiency for the random assignment setting. They characterize \strong-efficiency. In this note, we generalize the characterization to the more general voting setting. 

\section{Preliminaries}

Consider the social choice setting in which there is a set of agents $N=\{1,\ldots, n\}$, a set of alternatives $A=\{a_1,\ldots, a_m\}$ and a preference profile $\pref=(\pref_1,\ldots,\pref_n)$ such that each $\pref_i$ is a complete and transitive relation over $A$.
	We write~$a \pref_i b$ to denote that agent~$i$ values alternative~$a$ at least as much as alternative~$b$ and use~$\spref_i$ for the strict part of~$\pref_i$, i.e.,~$a \spref_i b$ iff~$a \pref_i b$ but not~$b \pref_i a$. Finally, $\indiff_i$ denotes~$i$'s indifference relation, i.e., $a \indiff_i b$ iff both~$a \pref_i b$ and~$b \pref_i a$.
The alternatives in $A$ could be any discrete structures: voting outcomes, house allocation, many-to-many two-sided matching, or coalition structures.	
A utility profile $u=(u_1,\ldots, u_n)$ specified for each agent $i\in N$ his utility for $u_i(a)$ for each alternative $a\in N$. A utility profile is \emph{consistent} with the preference profile $\pref$, if for each $i\in N$ and $a,b\in A$, $u_i(a)\geq u_i(b)$ if $a\pref_i b$.
Two alternatives $a,b\in A$ are \emph{Pareto indifferent} if $a\sim_i b$ for all $i\in N$. For any alternative $a\in A$, we will denote by $D(a)$ the set $\{b\in A\midd \exists i\in N, a\succ_i b\}$.

We will also consider randomized outcomes that are lotteries over $A$.
A lottery is a probability distribution over $A$. We denote the set of lotteries by $\Delta(A)$. For a lottery $p\in \Delta(A)$, we denote by $p(a)$ the probability of alternative $a\in A$ in lottery $p$. 
We denote by $supp(p)$ the set $\{a\in A\midd p(a)>0\}$.
		A lottery $p$ is \emph{interesting} if  there exist $a,b\in supp(p)$ such that there exist $i,j\in N$ such that $a\succ_i b$ and $b\succ_j a$. A lottery is \emph{degenerate} if it puts probability one on a single alternative.

		Under \emph{stochastic dominance (SD)}, an agent prefers a lottery that, for each alternative $x \in A$, has a higher probability of selecting an alternative that is at least as good as $x$. Formally, $p \pref_i^{\sd} q$ iff $\forall y\in\nolinebreak A \colon \sum_{x \in A: x \pref_i y} p(x) \geq \sum_{x \in A: x \pref_i y} q(x).$ 
It is well-known that $p \pref^{\sd} q$ iff $p$ yields at least as much expected utility as $q$ for any von-Neumann-Morgenstern utility function consistent with the ordinal preferences \cite{ABBH12a,Cho12a}.	
A lottery is \emph{SD-efficient} if it is Pareto optimal with respect to the SD relation.


\section{\strong-efficiency}

We now consider \strong-efficiency as introduced by  \citet{DDY15a}.
Although \citet{DDY15a} defined \strong-efficiency in the context of random assignment, the definition extends in a straightforward manner to the case of voting. 

\begin{definition}[SW-efficiency]
	A lottery $p$ is \emph{\strong-efficient} if there exists no other lottery $q$ that \strong dominates it. Lottery $q$ \strong dominates $p$ if 
for any utility profile for which $p$ maximizes welfare, $q$ maximises welfare, and there exists  at least one utility profile for which $q$ maximised welfare but $p$ does not. 
\end{definition}

\begin{lemma}\label{lemma:D}
	Consider a Pareto optimal alternative $a\in A$ and a non-empty set $D(a)=\{b\in A\midd \exists i\in N, a\succ_i b\}$. Then, there exists a utility profile $u$ such that $\sum_{i\in N}u_i(a)>\sum_{i\in N}u_i(b)$ for all  $b\in D(a)$.
		\end{lemma}
		\begin{proof}
			We can construct the require utility function profile $u$ as follows.
Whenever $a\succ_i b$, make the difference $u_i(a)-u_i(b)$ huge. Whenever $b\succ_j a$, make the difference $u_j(b)-u_j(a)$ arbitrarily small. Hence the value $u_i(a)-u_i(b)$ is large enough that it makes up for all $j$ for which $u_j(b)-u_j(a)>0$. Hence $\sum_{i\in N}u_i(a)-u_i(n)>0$. 
			\end{proof}

		\begin{lemma}\label{lemma:effrelations}
\strong-efficiency implies SD-efficiency implies ex post efficiency. 
			\end{lemma}
			\begin{proof}
				It is well-known that SD-efficiency implies ex post efficiency~\citep{ABBH12a}. 
				
				Consider a lottery $p$ that is not SD-efficient. Then there exists another lottery $q$ that SD-dominates it. Hence $p$ does not maximize welfare for any consistent utility profile because $q$ yields more utility for each utility profile. 
				\end{proof}

			\begin{lemma}\label{lemma:neccA}
				An interesting lottery is not \strong-efficient.
				\end{lemma}
				\begin{proof}
					If an interesting lottery $p$ is not SD-efficient, we are already done because by Lemma~\ref{lemma:effrelations}, $p$ is not \strong-efficient. So let us assume $p$ is SD-efficient and hence ex post efficient. Since $p$ is interesting, there exists at least one $a\in supp(p)$ such that $a\succ_i b$ for some $b\in supp(p)$ and $i\in N$. Note that $a$ is Pareto optimal. By Lemma~\ref{lemma:D}, there exists a utility profile $u$ such that $\sum_{i\in N}u_i(a)>\sum_{i\in N}u_i(b)$ for all  $b\in D(a)$ where $D(a)\cap supp(p)\neq \emptyset$. 
Hence, there exists a utility profile $u$ such that $\sum_{i\in N}u_i(a)>\sum_{i\in N}u_i(b)$ for all  $b\in supp(p)\cap D(a)$. This means that $\sum_{i\in N}u_i(a)> \sum_{i\in N}u_i(p)$.
					\end{proof}

						
						\begin{lemma}\label{lemma:suff}
							An uninteresting lottery over Pareto optimal alternatives is \strong-efficient.
							\end{lemma}
							\begin{proof}
								An uninteresting lottery $p$ over Pareto optimal alternatives is SD-efficient. Assume that there is another lottery $q$ that \strong-dominates $p$. Then $supp(q)$ contains one alternative $b$ that is not Pareto indifferent to alternatives $supp(p)$. This means that there exists a utility profile $u$ such that welfare is maximized by $p$ but not by $b$. Hence $q$ does not \strong-dominate $p$.
								\end{proof}

								\begin{theorem}\label{th:main}
									 A lottery is \strong-efficient iff it is ex post efficient and uninteresting.
									\end{theorem}
									\begin{proof}
										By Lemma~\ref{lemma:suff}, an ex post efficiency and uninteresting lottery is \strong-efficient. 
										
										We now prove that if lottery is not ex post efficient or uninteresting, it not \strong-efficient. 
Due to Lemma~\ref{lemma:effrelations}, if a lottery is not ex post efficient, it is not \strong-efficient. Similarly, by Lemma~\ref{lemma:neccA}, if a lottery is interesting, it is not \strong-efficient. 
										\end{proof}

	%
	%
	%


	\begin{lemma}\label{lemma:degen}
		If $A$ contains no Pareto indifferent alternatives, then 
	if a lottery is uninteresting and not degenerate, then it is not ex post efficient.	
		\end{lemma}
		\begin{proof}
			Assume that a lottery $p$ is uninteresting and not degenerate. 
			Since $p$ is not degenerate, $|supp(p)|\geq 2$. Since $p$ is uninteresting, there do not exist $a,b\in supp(p)$ such that there exist $i,j\in N$ such that $a\succ_i b$ and $b\succ_j a$. Thus either $a$ Pareto dominates $b$, or $b$ or Pareto dominates $a$ or $a$ and $b$ are Pareto indifferent. The third case is not possible because we assumed that $A$ does not contain Pareto indifferent alternatives. Since $a$ Pareto dominates $b$ or $b$ or Pareto dominates $a$, $supp(p)$ contains a Pareto dominated alternative. Hence $p$ is not ex post efficient.
					\end{proof}

	\begin{theorem}
		If $A$ contains no Pareto indifferent alternatives, then  a lottery is \strong-efficient iff it is ex post efficient and degenerate.
		\end{theorem}
		
	\begin{proof}
		Assume that $A$ contains no Pareto indifferent alternatives.
		If a lottery $p$ is \strong-efficient, then by Theorem~\ref{th:main}, it is ex post efficient and uninteresting. By Lemma~\ref{lemma:degen}, since $p$ is ex post efficient, it is degenerate. 
		
		Now assume that a lottery $p$ is ex post efficient and degenerate. Since $p$ is degenerate, it is uninteresting by definition. Since it is both ex post efficient and uninteresting, then by Theorem~\ref{th:main}, it is \strong-efficient.
		\end{proof}

		\begin{lemma}
			An assignment problem with strict preferences does not admit Pareto indifferent discrete assignments. 
			\end{lemma}
			\begin{proof}
				Consider two discrete assignments $M$ and $M'$ such that all agents are indifferent among them. Then this means that each agent gets the same item in both $M'$ and $M'$. But this implies that $M'=M$.
				\end{proof}
				
				\begin{corollary}[\citet{DDY15a}]
			If preferences are strict, the only undominated probabilistic assignments are the Pareto efficient deterministic assignments. 
								\end{corollary}
\begin{proof}
	No two discrete assignments are completely indifferent for all agents. Hence only random assignment that is \strong-efficient is a discrete Pareto optimal assignment. 
	\end{proof}

%
%
%
%

		\subsection*{Acknowledgments}
		NICTA is funded by the Australian Government through the Department of Communications and the Australian Research Council through the ICT Centre of Excellence Program.


\end{document}